\documentclass[10pt]{article}
\usepackage{amsfonts,amsthm,amsmath,amssymb,mathrsfs,url, braket, graphicx}
\usepackage[usenames,dvipsnames]{color}

\title{Detecting Wave Function Collapse Without Prior Knowledge}
\author{
Charles Wesley Cowan\footnote{Department of Mathematics, Rutgers University, Hill Center,  
     110 Frelinghuysen Road, Piscataway, NJ 08854-8019, USA.}\ \footnote{E-mail:
     cwcowan@math.rutgers.edu}\ \ and 
Roderich Tumulka$^*$\footnote{E-mail: tumulka@math.rutgers.edu}
}
\date{July 23, 2015}

\newcommand{\Hilbert}{\mathscr{H}}
\newcommand{\E}{\mathcal{E}}
\newcommand{\EEE}{\mathbb{E}}
\newcommand{\PPP}{\mathbb{P}}
\newcommand{\RRR}{\mathbb{R}}
\newcommand{\CCC}{\mathbb{C}}
\newcommand{\SSS}{\mathbb{S}}

\newcommand{\NNN}{\mathbb{N}}
\newcommand{\scp}[2]{\langle #1|#2 \rangle}
\newcommand{\pr}[1]{| #1 \rangle \langle #1 |}
\newcommand{\be}{\begin{equation}}
\newcommand{\ee}{\end{equation}}

\DeclareMathOperator{\tr}{tr}
\DeclareMathOperator{\diag}{diag}
\theoremstyle{plain}
	\newtheorem{prop}{Proposition}
	\newtheorem{theorem}{Theorem}
	
	\newtheorem{corollary}{Corollary}
	\newtheorem{conjecture}{Conjecture}

\begin{document}
\maketitle

\begin{abstract}
We are concerned with the problem of detecting with high probability whether a wave function has collapsed or not, in the following framework: A quantum system with a $d$-dimensional Hilbert space is initially in state $\psi$; with probability $0<p<1$, the state collapses relative to the orthonormal basis $b_1,\ldots,b_d$. That is, the final state $\psi'$ is random; it is $\psi$ with probability $1-p$ and $b_k$ (up to a phase) with $p$ times Born's probability $|\scp{b_k}{\psi}|^2$. Now an experiment on the system in state $\psi'$ is desired that provides information about whether or not a collapse has occurred. Elsewhere \cite{CT12a}, we identify and discuss the optimal experiment in case that $\psi$ is either known or random with a known probability distribution. Here we present results about the case that no \textit{a priori} information about $\psi$ is available, while we regard $p$ and $b_1,\ldots,b_d$ as known. For certain values of $p$, we show that the set of $\psi$s for which any experiment $\E$ is more reliable than blind guessing is at most half the unit sphere; thus, in this regime, any experiment is of questionable use, if any at all. Remarkably, however, there are other values of $p$ and experiments $\E$ such that the set of $\psi$s for which $\E$ is more reliable than blind guessing has measure greater than half the sphere, though with a conjectured maximum of 64\% of the sphere.

\medskip

\end{abstract}

\section{Introduction}\label{sec:introduction}

We consider a quantum system whose wave function may or may not have collapsed, and ask to what extent  any experiment on the system can provide us with information about whether it has collapsed. The answer depends on how much is known about the initial wave function, and we focus here on the case that the initial wave function is completely unknown, with no prior assumptions as to its state. 

Our motivation comes from ``spontaneous-collapse'' theories of quantum mechanics, such as that of Ghirardi, Rimini, and Weber (GRW) \cite{GRW86,GTZ07}, according to which collapses of the wave function occur as objective physical events, so it is of interest to detect these events when they occur, and it is of interest to note that no experiment is able to detect a collapse with 100\%\ reliability, even when the initial wave function is known; see \cite{CT12b} for a discussion of the implications of our results for the GRW theory. Yet, our results are not limited to the GRW theory but apply equally to the collapses considered in orthodox quantum theory.

The problem can be formulated mathematically as follows. Consider a quantum system $S$ with Hilbert space $\Hilbert$ of finite dimension $d \in \NNN$, $d \geq 2$. Let
\begin{equation}\label{eqn:SSSdef}
\SSS = \{\psi \in \Hilbert: \|\psi\|=1\}
\end{equation}
denote the unit sphere in $\Hilbert$, and let $B = \{b_1, \ldots, b_d\}$ be an orthonormal basis of $\Hilbert$. Suppose that the ``initial'' wave function of $S$ was $\psi\in\SSS$ but with probability $p$ a collapse relative to $B$ has occurred. That is, suppose that the wave function of $S$ is the $\SSS$-valued random variable $\psi'$ defined to be
\begin{equation}\label{eqn:psi-prime-def}
\psi' = \begin{cases}
\psi & \text{with probability }1-p\\
\frac{\scp{b_k}{\psi}}{|\scp{b_k}{\psi}|}b_k & \text{with probability }p\,\bigl| \scp{b_k}{\psi} \bigr|^2\text{ for }k=1,\ldots,d\,.
\end{cases}
\end{equation}
We take $p$ and $B$ to be known;\footnote{Actually, the problem depends on $B$ only through its equivalence class, with the basis $\{e^{i\theta_1}b_1,\ldots,e^{i\theta_d}b_d\}$ regarded as equivalent to $B$ for arbitrary $\theta_1,\ldots,\theta_d\in\RRR$. So we take the equivalence class of $B$ (or, equivalently, the collection of $d$ 1-dimensional subspaces $\CCC b_k$) to be known; nevertheless, we often find it convenient to speak as if $B$ were given.} $\psi$ may or may not be known.

Now we desire an experiment on $S$ that reveals, or at least provides probabilistic information about, whether a collapse has occurred. A relevant fact is that for every conceivable experiment $\E$ that can be carried out on $S$ (with random outcome $Z$ from some value space $\mathcal{Z}$), there is a positive-operator-valued measure (POVM) $M(\cdot)$ on $\mathcal{Z}$ acting on $\Hilbert$ such that the probability distribution of $Z$, when $\E$ is carried out on $S$ with wave function $\psi\in\SSS$, is given by
\begin{equation}\label{eqn:POVM-probability-statement}
\PPP(Z\in\Delta) = \scp{\psi}{M(\Delta)|\psi}
\end{equation}
for all measurable sets $\Delta\subseteq\mathcal{Z}$. 
The statement containing \eqref{eqn:POVM-probability-statement} was proved for GRW theory in \cite{GTZ07} and for Bohmian mechanics in \cite{DGZ04}. In orthodox quantum mechanics, which does not allow for an analysis of measurement processes, it cannot be proved rigorously, but it can be derived from the assumption that, after $\E$, a quantum measurement of the position observable of the pointer of $\E$'s apparatus will yield the result of $\E$. However, it is important to note that while every experiment $\E$ can be characterized in terms of a POVM, it is not necessarily true that every POVM is associated with a realizable experiment.  

For our purposes, so as to answer the question ``Did collapse occur?'', it suffices to consider \emph{yes-no} experiments, i.e., those with $\mathcal{Z}=\{\text{yes},\text{no}\}$; for them the POVM $M(\cdot)$ is determined by the operator
\begin{equation}\label{eqn:povm-yes-operator}
E=M(\{\text{yes}\})\,.
\end{equation}
Writing $I$ for the identity operator, $I-E$ is the operator corresponding to $\text{no}$, $M(\{\text{no}\})$. By the definition of a POVM, $E$ must be a positive\footnote{We take the word ``positive'' for an operator to mean $\braket{\psi|E|\psi}\geq 0$ for all $\psi\in\Hilbert$, equivalently to its matrix (relative to any orthonormal basis) being positive semi-definite; we  denote this by $E\geq 0$.} operator such that $I-E$ is positive too; it is otherwise arbitrary. Thus, \emph{we can characterize every possible yes-no experiment $\E$ mathematically by a self-adjoint operator $E$ with spectrum in $[0,1]$, $0 \leq E \leq I$.} As noted, this is a larger set than the class of ``realizable'' experiments, but by proving results over the set of POVMs (in this case of yes-no experiments, proving results over the set of self-adjoint operators with the appropriate spectrum), the results necessarily cover all possible realizable experiments.

We define the \emph{reliability} $R(\E)$ of a yes-no experiment $\E$ to be the probability that it will, when applied to the system $S$ in the state $\psi'$, correctly retrodict whether a collapse has occurred. Since $R(\E)$ depends on $\E$ only through $E$, we also write $R(E)$ for it. We use this quantity to measure how useful an experiment is for our purposes---how much an experimenter might ``rely'' on the experiment's outcome. It is a basic fact \cite{CT12a} that perfect reliability is impossible; i.e., always $R(\E)<1$. Values of reliability should be compared to that of \emph{blind guessing}, the trivial experiment $\E_\emptyset$ that always yields the outcome ``yes'' if $p>1/2$ and always ``no'' if $p\leq 1/2$.

Elsewhere \cite{CT12a}, we discuss the case in which $\psi$ is either known or random with a known probability distribution; we briefly report the main results in Sec.~\ref{sec:report} below. In the central results of this paper, we take $\psi$ to be unknown, with no prior information. This case is relevant, e.g., if we consider a world governed by the GRW theory whose inhabitants want to detect the collapses during a given time interval in a given system $S$; for them, a machine would not qualify as a ``collapse detector'' unless it works irrespectively of the initial state of $S$.

It is not obvious what it should mean for an experiment to work for unknown $\psi$. One approach, with a Bayesian flavor, would be to take this to mean that the experiment is more reliable than blind guessing for a random $\psi$ from a \emph{uniform distribution}; see Sec.~\ref{sec:report} or \cite{CT12a} for results. However, we question whether an unknown $\psi$ can truly be assumed to be uniformly distributed. For instance, if the experimenter were to receive only a single $\psi$ to experiment on, what would its ``distribution'' really mean? So we follow another approach. Obviously, any experiment $\E$ that we choose will have high reliability for some $\psi$s and low reliability for others. We may therefore wish to maximize the size of the set
\be
S_\E = \bigl\{ \psi\in\SSS(\Hilbert): R_\psi(\E)>R_\psi(\E_\emptyset) \bigr\}\,,
\ee
the set of $\psi$s for which $\E$ is more reliable than blind guessing $\E_\emptyset$; $R_\psi(\E)$ denotes the reliability of $\E$ for fixed $\psi$. The natural measure of ``size'' is the (normalized) uniform measure $u$ on $\SSS$.

We thus ask, how big can we make $u(S_\E)$ by suitable choice of $\E$? If we can bring it close to 1 then that experiment might be quite attractive. If, in contrast, $u(S_\E)\leq 1/2$ then $\E$ appears to be useless. We will largely answer this question, though not completely. The answer depends on the dimension $d$ and the value of $p$. We show, among other things, that for $d=2$ or $p$ close to 0 or 1, $u(S_\E)\leq 1/2$ for all $\E$. Furthermore, while there are operators $0\leq E\leq I$ with $u(S_E)>1/2$ in dimension $d\geq 3$, we present reasons to conjecture that always $u(S_E)\leq 1-1/e\approx0.632$. Since that value is not particularly close to 1, we conclude that, without prior information about $\psi$, collapses cannot be detected in a useful way.

The rest of the paper is organized as follows. In Sec.~\ref{sec:report}, we summarize the results of \cite{CT12a} about the case that $\psi$ is known or random with known distribution. In Sec.~\ref{sec:no-information-experimentation}, we present our new results. Long proofs and calculations relevant to Sec.~\ref{sec:no-information-experimentation} are mostly postponed to Sec.~\ref{sec:proofs}.



\section{If $\psi$ is Known or Random with Known Distribution}
\label{sec:report}

Assuming that the initial state $\psi$ is known, the reliability is found to be \cite{CT12a}
\begin{equation}\label{eqn:reliability-psi-formula}
R_{\psi, p}(\E) = p \braket{\psi | \diag E | \psi } + (1-p) \braket{\psi | I - E | \psi},
\end{equation}
where
\be\label{diagdef}
\diag E = \sum_{k = 1}^d \ket{b_k}\braket{b_k|E|b_k} \bra{b_k}
\ee
is the ``diagonal part'' of the operator $E$ relative to the basis $B$. In particular, the reliability depends on $\E$ only through $E$. A perhaps surprising result \cite{CT12a} is that, for $p \geq d/(d+1)$, no experiment is more reliable than blind guessing.

If $\psi$ is random with known distribution $\mu$ on the unit sphere $\SSS$ in Hilbert space $\Hilbert$, then the reliability $R_{\mu,p}(\E)$ of $\E$ is the average of $R_{\psi,p}(\E)$ over $\psi$, 
\begin{equation}\label{eqn:Rmu=ERpsi}
R_{\mu,p}(\E) = \int_{\SSS}\mu(d\psi)  \, R_{\psi,p}(\E) \,,
\end{equation}
and more explicitly given by
\begin{equation}\label{eqn:limited-information-reliability-eqn}
R_{\mu, p}(\E) = \tr \Bigl[ \rho \left( p\diag E + (1-p)(I-E) \right) \Bigr]\,,
\end{equation}
where $\rho$ is the density matrix associated with distribution $\mu$, defined by
\begin{equation}\label{eqn:rho-definition-eqn}
\rho = \int_{\SSS} \mu(d\psi) \, \pr{\psi}\,.
\end{equation}
In particular, the reliability depends on $\mu$ only through $\rho$ (and on $\E$ only through $E$).

The problem of detecting collapse can be paraphrased as the problem of distinguishing between the two density matrices $\rho_1=\rho$ (corresponding to the distribution $\mu$ of the pre-collapse state $\psi$) and $\rho_2=\diag \rho$ (which is the density matrix corresponding to the distribution of the post-collapse state). That is, we are given, with probability $p$, a system with density matrix $\rho_2$ or, with probability $1-p$, a system with density matrix $\rho_1$, and want to decide which of the two cases has occurred. For any $\{1,2\}$-valued experiment $\E$, its \emph{reliability} $R_{\rho_1,\rho_2,p}(\E)$ is defined to be the probability of giving the correct retrodiction. It is found to be, in terms of the POVM $\{E_1,E_2=I-E_1\}$ of $\E$,
\begin{equation}\label{RHelstrom}
R_{\rho_1, \rho_2, p}(\E) 
= 1-p + \tr \left[ A E_1 \right]
\end{equation}
with
\begin{equation}\label{A-def-general}
A = p \rho_1 - (1-p) \rho_2\,.
\end{equation}
The operator $E_1$ that maximizes the reliability was identified by Helstrom \cite{Hel76} as follows:

\begin{prop}\label{prop:rho1rho2}
For $0 \leq p \leq 1$ and any density matrices $\rho_1, \rho_2$,
\begin{equation}\label{general-optimality}
R^{\max}_{\rho_1,\rho_2,p}:=\max_{0\leq E_1\leq I} R_{\rho_1,\rho_2,p}(E_1) = (1-p) + \lambda^+ = p - \lambda^-,
\end{equation}
where $\lambda^+\geq 0$ and $\lambda^-\leq 0$ are, respectively, the sum of the positive eigenvalues (with multiplicities) and that of the negative eigenvalues of $A$ as in \eqref{A-def-general}.
The optimal operators $E_1=E_\mathrm{opt}$ for which this maximum is attained, $R^{\max}_{\rho_1, \rho_2,p} = R_{\rho_1, \rho_2, p}(E_\mathrm{opt})$, are those satisfying
\be\label{Eoptrho12}
P^+_A \leq E_\mathrm{opt} \leq P^+_A + P^0_A\,,
\ee
where $P^+_A$ is the projection onto the positive spectral subspace of $A$, i.e, onto the sum of all eigenspaces of $A$ with positive eigenvalues, and $P^0_A$ is the projection onto the kernel of $A$.
\end{prop}

Let us return to the special case of detecing collapse ($\rho_1=\rho,\rho_2=\diag \rho$). A particular distribution $\mu$ of interest is the (normalized) uniform distribution $u$ over the unit sphere; for it, $\rho=d^{-1}I = \diag \rho$, and it is impossible to distinguish between the two density matrices. In fact, independently of $p$, no experiment can provide any additional information at all for $\mu=u$ about whether collapse has occurred. In this setting, blind guessing is  always optimal.

The other extreme special case is that of known $\psi$ ($\mu=\delta_\psi, \rho=\pr{\psi}$). We obtain from Prop.~\ref{prop:rho1rho2} that blind guessing is optimal for $p\geq d/(d+1)$, while for $p<d/(d+1)$, the optimal operator $E$ is the projection to the $(d-1)$-dimensional subspace of $\Hilbert$ orthogonal to the unique (up to a factor) eigenvector $\phi$ of $A$ as in \eqref{A-def-general} with a negative eigenvalue; $\phi$ is computed explicitly in \cite{CT12a}, along with the maximal value of reliability.

Instead of the exact value of $R^{\max}$, it is often more practical to compute bounds on $R^{\max}$. We provide several such bounds in \cite{CT12a}, the simplest of which is
\be
R^{\max}_{\rho,p} \leq \max\Bigl(p,1-\frac{p}{d}\Bigr)\,.
\ee
In fact, for $p\geq d/(d+1)$, $\max(p,1-p/d) = p$, which is the reliability of blind guessing.


\section{No Prior Information About $\psi$}\label{sec:no-information-experimentation}

For a particular experiment $\E$, for what fraction of $\SSS$ does $\E$ perform better than blind guessing? That is, how big is the set $S_\E$ where $R_\psi(\E)$ surpasses the reliability $R(\E_\emptyset)$ of blind guessing? The results reported in the previous section imply that the average of $R_\psi(\E)$ over $\psi$ (uniformly distributed) is at most $R(\E_\emptyset)$. As such, values above $R(\E_\emptyset)$ in $S_\E$ come at the price of values below $R(\E_\emptyset)$, outside $S_\E$; but that might be acceptable if $S_\E$ is very large.

Let $\mu$ be the normalized uniform measure on $\SSS$ (called $u$ so far). Define $\Lambda_p(E)$ to be the measure of the set of $\psi$ for which $E$ is more reliable than blind guessing. That is, with $E_\emptyset$ the operator associated with blind guessing (i.e., $E_\emptyset=0$ for $p\leq 1/2$ and $E_\emptyset=I$ for $p> 1/2$),
\begin{equation}\label{eqn:lambda-def}
\Lambda_p(E) = \mu\Bigl\{ \psi \in \SSS : R_{\psi, p}(E) > R_{\psi, p}(E_\emptyset) \Bigr\}.
\end{equation}
This definition is equivalent to $\Lambda_p(E) = \PPP\bigl( R_{\psi, p}(E) > R_{\psi, p}(E_\emptyset) \bigr)$ when $\psi$ is sampled from $\SSS$ with the uniform distribution. This probabilistic interpretation of $\Lambda_p(E)$ will let us calculate many things explicitly.

For instance, we offer the following bound, which we will proceed to tighten.

\begin{prop}\label{prop:markov-bound}
For $p \neq 1/2$ and any $0\leq E\leq I$, $\Lambda_p(E) < 1$. That is, the set where $E$ is more reliable than blind guessing is strictly smaller than all $\SSS$.
\end{prop}

\begin{proof}
Taking the probabilistic view as mentioned, and applying Markov's inequality,
\begin{equation}
\begin{split}
\Lambda_p(E) & = \PPP\Bigl( R_{\psi, p}(E) > R_{\psi, p}(E_\emptyset) \Bigr) \\
& = \PPP\Bigl( R_{\psi, p}(E) > \max(p,1-p) \Bigr) \\
& \leq \frac{ \EEE [R_{\psi, p}(E)] }{  \max(p,1-p) },
\end{split}
\end{equation}
with the expectation taken with respect to the uniform distribution over $\SSS$. Since this expectation is $(1-p) + (2p-1) \tr E/d$,
\begin{equation}\label{eqn:markov-bound}
\begin{split}
\Lambda_p(E) & \leq \frac{ (1-p) + (2p-1) \tr E / d }{  \max(p, 1-p) } \\
& = \frac{ (1-p) + (2p-1) \tr E / d }{  (1-p) + \max(0, 2p-1) }.
\end{split}
\end{equation}
For non-trivial $E$, and $p \neq 1/2$, the above gives $\Lambda_p(E) < 1$.
\end{proof}
For any particular $E$, \eqref{eqn:markov-bound} can give a much better bound, but this serves to show that for any non-trivial $E$, $p \neq 1/2$, there is a set of $\psi$ of positive measure for which $E$ is no more reliable than blind guessing. For $p = 1/2$, \eqref{eqn:markov-bound} gives the unenlightening upper bound of $1$ independent of $E$, but we will explore this in more depth shortly.

Note that $R_{\psi, p}(E)$ may be expressed equivalently as
\begin{equation}
\begin{split}
R_{\psi, p}(E) & = \begin{cases}
(1-p) + \braket{\psi| p \diag E - (1-p) E |\psi} & \text{if } p \leq 1/2 \\
p + \braket{\psi| (1-p)(I-E) - p \diag(I - E)| \psi} & \text{if } p > 1/2\\
\end{cases} \\
& = \max(p, 1-p) + \braket{\psi|A_p(E)|\psi},
\end{split}
\end{equation}
where $A_p(E)$ is defined by
\begin{equation}
\label{eqn:ApE-definition}
A_p(E) = \begin{cases}
p \diag E - (1-p) E & \text{if } p \leq 1/2 \\
(1-p)(I - E) - p \diag(I - E) & \text{if } p > 1/2. \\
\end{cases}
\end{equation}
Hence, we may equivalently express the condition that $R_{\psi, p}(E) > R_{\psi, p}(E_\emptyset)$ as
\be\label{psiApsi0}
\braket{\psi|A_p(E)|\psi} > 0\,.
\ee
That is, $E$ is more reliable than blind guessing for a given $\psi$ iff \eqref{psiApsi0} holds. As a function of $\psi$, $\braket{\psi|A_p(E)|\psi}$ is maximized and minimized at the eigenvectors corresponding to the largest and smallest eigenvalues, respectively. If the eigenvalues of $A_p(E)$ are all at most $0$, then $\braket{\psi|A_p(E)|\psi} > 0$ is never satisfied, and $\Lambda_p(E) = 0$.

In fact, we may express $\Lambda_p(E)$ precisely in terms of the eigenvalues of $A_p(E)$, via the following theorem.

\begin{theorem}\label{thm:measure-formula}
Consider a self-adjoint $d\times d$ matrix $A$ with non-degenerate positive eigenvalues given by $\alpha_1 > \alpha_2 > \ldots > \alpha_k > 0$, and non-positive eigenvalues given by $0 \geq \beta_1 \geq \beta_2 \geq ... \geq \beta_m$, with $k + m = d$ and $k,m\geq 1$. Then
\begin{equation}
\label{eqn:set-measure-formula}
\mu \bigl\{ \psi \in \SSS : \braket{\psi|A|\psi} > 0 \bigr\} = \sum_{i = 1}^k \frac{\alpha_i^{d-1}}{\prod_{h = 1}^m(\alpha_i - \beta_h) \prod_{j = 1, j \neq i}^k(\alpha_i - \alpha_j)},
\end{equation}
where $\mu$ is the normalized, uniform measure on $\SSS$.
\end{theorem}

The proof is given in Sec.~\ref{sec:proofs:lambda}.

Thm.~\ref{thm:measure-formula}  gives us an explicit formula for $\Lambda_p(E)$, taking $A = A_p(E)$. 
We note for later use that
\begin{equation}\label{eqn:trace-equation}
\tr A_p(E) = \begin{cases}
 -(1-2p)\tr E & \text{if } p \leq 1/2 \\
 -(2p-1) \tr \bigl[ I - E \bigr] & \text{if } p > 1/2.
\end{cases}
\end{equation}
In particular, always $\tr A_p(E)\leq 0$.

Finding $E$ to maximize $\Lambda_p(E)$ can be treated as a problem of maximizing \eqref{eqn:set-measure-formula} as a function of the $\alpha_i, \beta_j$. However, it is not obvious which sets of $\alpha_i, \beta_j$ are realizable as eigenvalues of $A_p(E)$ for $0\leq E\leq I$ and $0 < p < 1$. Some results characterizing the spectrum of $A_p(E)$ are given in Sec.~\ref{sec:proofs:spectrum}. A weaker but still interesting problem would be to maximize \eqref{eqn:set-measure-formula} given the trace constraint that
\begin{equation}
 \sum_{i = 1}^k \alpha_i + \sum_{j = 1}^m \beta_j = \tr A_p(E)\,.
\end{equation}
This would give an upper bound for $\Lambda_p(E)$. However, in general it seems difficult to analytically find a global maximum of the function given the constraints.

The case in two dimensions is fairly immediate, however.

\begin{theorem}
\label{thm:2d-guessing-set}
For $d = 2$, any operator $0\leq E\leq I$, and any $0 < p < 1$, $\Lambda_p(E) \leq 1/2$.
\end{theorem}

\begin{proof}
At least one eigenvalue of $A_p(E)$ must be non-positive because $\tr A_p(E)\leq 0$. 
If both eigenvalues of $A_p(E)$ are non-positive, $\Lambda_p(E) = 0$. If, however, one is positive, let $\alpha > 0 > \beta$ be the eigenvalues of $A_p(E)$. Then from Eq.~\eqref{eqn:set-measure-formula}, 
\begin{equation}
\Lambda_p(E) =  \frac{\alpha}{\alpha - \beta}\,.
\end{equation}
Noting that $\alpha + \beta = \tr A_p(E)$, $-\beta = \alpha - \tr A_p(E)$. Hence, $\Lambda_p(E) = \frac{\alpha}{2\alpha - \tr A_p(E)}$. Since $\tr A_p(E)\leq 0$, $\Lambda_p(E) \leq 1/2$.
\end{proof}

Hence, in two dimensions, the set of $\psi$ for which any non-trivial experiment is more reliable than blind guessing is no greater than half the sphere. It would seem, therefore, that we have nothing to lose by blind guessing. In higher dimensions, the situation is somewhat improved. Some limits remain, however.

\begin{theorem}
\label{thm:chernoff-theorem}
For all operators $0\leq E\leq I$ and all $0\leq p\leq 1$,
\be\label{chernoff-bound}
\Lambda_p(E) \leq 4 p (1-p)\,.
\ee
\end{theorem}

The proof is given in Sec.~\ref{sec:proofs:lambda}. Interestingly, this bound is independent of $d$. This result also agrees with \eqref{eqn:markov-bound} and a rough intuition of how the reliability ought to behave: At the extremes, $p = 0$ or $p = 1$, where collapse is either forbidden or guaranteed, any experiment other than blind guessing introduces the possibility of an incorrect outcome, hence no experiment can do better than blind guessing anywhere on the sphere. The upper bounds provided here and by \eqref{eqn:markov-bound} both increase as $p$ approaches $1/2$, peaking at $1$ with $p = 1/2$. Of course, $\Lambda_p(E)$ is always at most 1, so the case of $p =1/2$ is of little interest here. In general, the bound \eqref{chernoff-bound} is quite weak and deserves improvement, but one relevant consequence is the following. 

\begin{corollary}\label{cor:sqrt8}
Immediately from Thm.~\ref{thm:chernoff-theorem} (see Fig.~\ref{fig3}), independently of $d$, if $p < 1/2 - 1/\sqrt{8} \approx 0.146$ or $p > 1/2 + 1/\sqrt{8} \approx 0.854$, no experiment is more reliable than blind guessing for more than half of $\SSS$.
\end{corollary}

\begin{figure}[h]
\begin{center}
\includegraphics[width=0.33\textwidth]{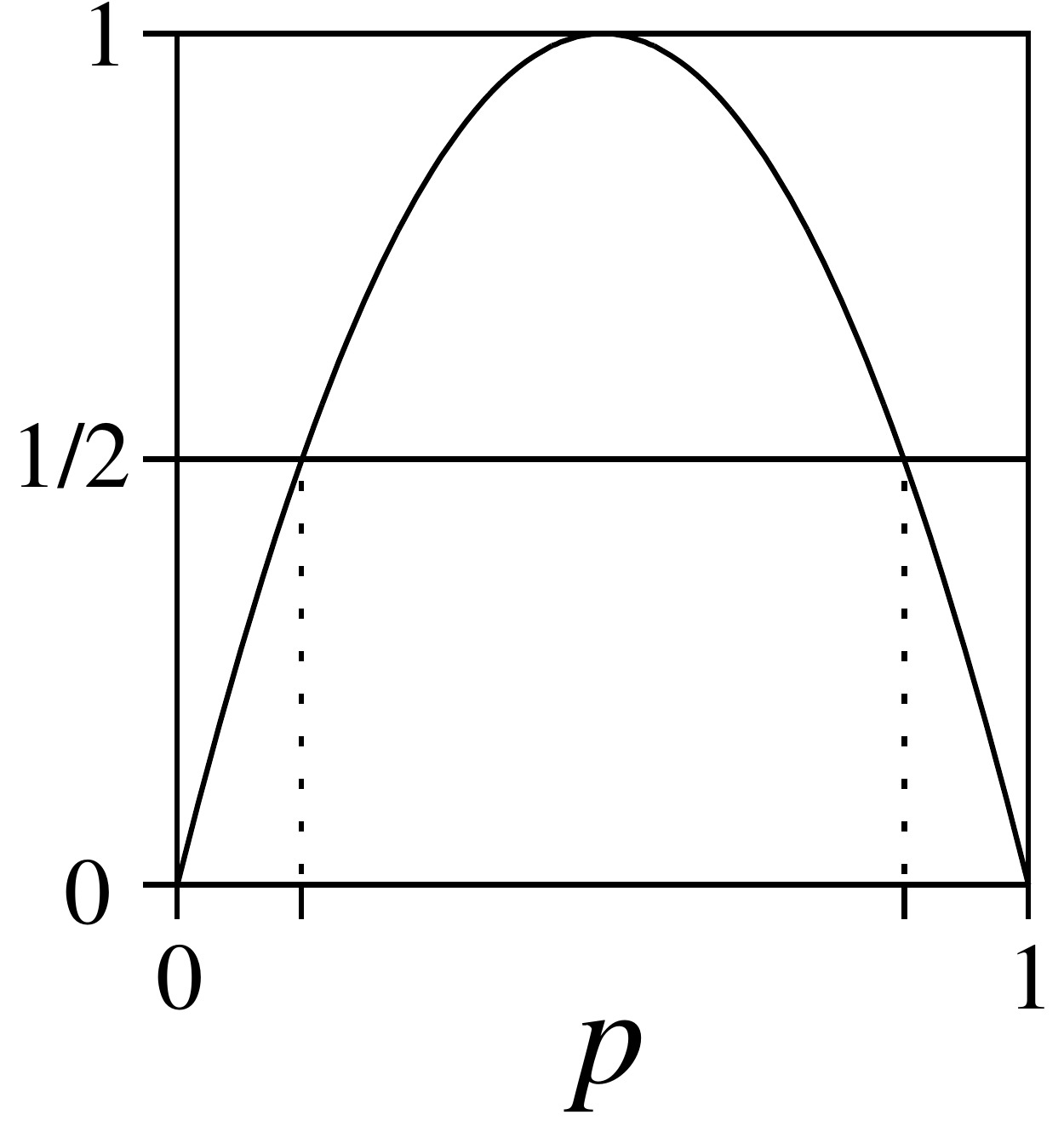}
\end{center}
\caption{
Graph of the function $f(p)=4p(1-p)$, relevant to Thm.~\ref{thm:chernoff-theorem} and Cor.~\ref{cor:sqrt8}. The marked points on the $p$-axis are at $1/2\pm1/\sqrt{8}$, which is where $f(p)=1/2$.}
\label{fig3}
\end{figure}

More can be said under assumptions about the distribution of eigenvalues of $A_p(E)$. Noting that $\tr A_p(E) < 0$ for $p\neq 1/2$ and $0\neq E\neq I$, $A_p(E)$ must have at least one negative eigenvalue. As noted, if $A_p(E)$ has all non-positive eigenvalues, $\Lambda_p(E) = 0$. We offer, therefore, the following extreme cases: that $A_p(E)$ has a single positive eigenvalue, and that $A_p(E)$ has a single negative eigenvalue. In either case, \eqref{eqn:set-measure-formula} is particularly easy to control.

\begin{theorem}
\label{thm:one-positive-eigenvalue}
For any $0 \leq E \leq I$ and $0<p<1$ for which $A_p(E)$ has only one positive eigenvalue, $\Lambda_p(E) \leq 1/2$.
\end{theorem}

\begin{proof}
Fix $E$ such that the eigenvalues of $A_p(E)$ are $\alpha > 0 \geq \beta_1 \geq \beta_2 \geq ... \geq \beta_{d-1}$. Under this condition, \eqref{eqn:set-measure-formula} gives
\begin{equation}
\begin{split}
\Lambda_p(E) & = \frac{\alpha^{d-1}}{\prod_{h = 1}^{d-1}(\alpha - \beta_h)} \\
& = \prod_{i = 1}^{d-1} \frac{1}{1 - \beta_i/\alpha}\,.
\end{split}
\end{equation}
It therefore suffices to show $\prod_{i = 1}^{d-1} (1 - \beta_i/\alpha) \geq 2$. Note that for each $i$, $-\beta_i/\alpha$ is non-negative. If $x, y \geq 0$, we have that $(1+x)(1+y) \geq 1 + x + y$. Generalizing this,
\be\label{prod1+xi}
\prod_{i=1}^{d-1} (1+x_i) \geq 1 + \sum_{i=1}^{d-1}x_i \quad \text{for }x_i\geq 0\,, 
\ee
and thus
\begin{equation}
\begin{split}
\prod_{i = 1}^{d-1} (1 - \beta_i/\alpha) & \geq 1 - \frac{1}{\alpha} \sum_{i = 1}^{d-1} \beta_i  \\
 & =  1 - \frac{1}{\alpha}(\tr A_p(E) - \alpha) \\ 
& =  2 - \tr A_p(E)/\alpha \\
& \geq 2.
\end{split}
\end{equation}
The last step follows since $\tr A_p(E)$ is non-positive.
\end{proof}

The situation is slightly improved in the alternate case.

\begin{theorem}
\label{thm:one-negative-eigenvalue}
For $p < \frac{\ln 2}{1 + \ln 2} \approx 0.409$ and $p > \frac{1}{1 + \ln 2} \approx 0.591$, for any operator $0 \leq E \leq I$ for which $A_p(E)$ has only one negative eigenvalue, $\Lambda_p(E) \leq 1/2$.
\end{theorem}

The proof appears in Sec.~\ref{sec:proofs:lambda}. Thm.~\ref{thm:one-negative-eigenvalue} is not as strong as Thm.~\ref{thm:one-positive-eigenvalue}, as the conclusion does not hold for all $p$. In fact, it is possible to find $E$ that are more reliable than blind guessing on greater than half of $\SSS$.

\begin{theorem}
\label{thm:counter-example}
For every $d \geq 3$, for some values of $p$, there exist operators $0 \leq E \leq I$ such that $\Lambda_p(E) > 1/2$.
\end{theorem}

\begin{proof}
Consider $E=\pr{\phi}$ with $\phi_k=1/\sqrt{d}$ for all $k=1,\ldots,d$. The eigenvalues of $E$ are $1$ and $0$, $0$ having multiplicity $d-1$. For $p \leq 1/2$, the eigenvalues of $A_p(E)$ are $p/d$ with multiplicity $d-1$, and $p/d - (1-p)$. In this case, we have a single negative eigenvalue.

For computational purposes, it is convenient to perform the following inversion and compute relative to $-A_p(E)$, which has a single positive eigenvalue: Since the set for which $\braket{\psi|A_p(E)|\psi} = 0$ has measure 0, we have that the measure of the set for which $\braket{\psi|A_p(E)|\psi} > 0$ and the measure of the set for which $\braket{\psi|A_p(E)|\psi} < 0$, or $\braket{\psi|-A_p(E)|\psi} > 0$, must sum to 1. Hence, by \eqref{eqn:set-measure-formula},
\begin{equation}\label{eqn:counter-example-calculation}
\begin{split}
\Lambda_p(E) & = 1 - \mu \Bigl\{ \psi \in \SSS : \braket{\psi|-A_p(E)|\psi} > 0 \Bigr\} \\
& =  1 - \frac{\bigl((1-p) - p/d\bigr)^{d-1}}{\prod_{h = 1}^{d-1}((1-p) - p/d  + p /d) } \\
& = 1 - \frac{\bigl((1-p) - p/d\bigr)^{d-1}}{(1-p)^{d-1} } \\
& = 1 - \left( 1 - \frac{p}{d(1-p)} \right)^{\!\! d-1}.
\end{split}
\end{equation}
From this, for $p \leq 1/2$, $\Lambda_p(E) > 1/2$ for
\begin{equation}
\label{eqn:good-p}
p > 1 - \frac{1}{1 + d\Bigl(1 - (\frac{1}{2})^{\frac{1}{d-1}}\Bigr)} \,.
\end{equation}
For $d = 3$, this gives $p \gtrsim 0.4677$, and the bound decreases as $d \rightarrow \infty$ to $\ln 2 /(1 + \ln 2) \approx 0.409$. Below this point, Thm.~\ref{thm:one-negative-eigenvalue} applies.
\end{proof}

This example shows that, at least for some values of $p$, it is possible to be more reliable than blind guessing on more than half of $\SSS$. The previous results give obvious bounds on this in terms of $p$ and $E$, but the question remains: How much better than blind guessing can $E$ be, over $\SSS$? What is the maximal value of $\Lambda_p(E)$?

Prior results, Thm.s~\ref{thm:one-negative-eigenvalue} and \ref{thm:chernoff-theorem} and Prop.~\ref{prop:markov-bound}, have implied a certain symmetry with respect to $p$, centered on $p = 1/2$, with something peculiar at $p = 1/2$. This is emphasized by the following result.

\begin{theorem}\label{thm:maximum}
For any operator $0 \leq E \leq I$, $\Lambda_p(E)$ is unimodal as a function of $p$, with a maximum at $p = 1/2$. Hence, $\Lambda_p(E) \leq \Lambda_{1/2}(E)$ for all $p, E$.
\end{theorem}

\begin{proof}
Let $q < p \leq 1/2$, and let $E$ be arbitrary in $[0, I]$. In this case, for any $\psi \in \SSS$,
\begin{equation}
\begin{split}
\braket{\psi|A_p(E)|\psi} - \braket{\psi|A_q(E)|\psi} & = (p-q) \braket{\psi|\diag E|\psi} + ((1-q) - (1-p)) \braket{\psi|E|\psi} \\
& = (p-q) \left( \braket{\psi|\diag E|\psi} + \braket{\psi|E|\psi} \right) \\
& \geq 0\,.
\end{split}
\end{equation}
The last step follows since $p > q$, and $\diag E$ and $E$ are positive operators. Hence,
\begin{equation}
\braket{\psi|A_p(E)|\psi} \geq \braket{\psi|A_q(E)|\psi}.
\end{equation}
As $\braket{\psi|A_p(E)|\psi}$ is non-decreasing with $p$ for $p \leq 1/2$, $\Lambda_p(E)$ is similarly non-decreasing with $p$ for $p \leq 1/2$. A similar argument holds for $p \geq 1/2$, showing $\braket{\psi|A_p(E)|\psi}$ is non-increasing. Hence, $\Lambda_p(E)$ is non-decreasing for $p < 1/2$, non-increasing for $p > 1/2$.
\end{proof}

This result suggests the following, which we offer as a conjecture.

\begin{conjecture}
\label{conj:that-damn-conjecture}
For all $d\geq 2$, all operators $0 \leq E \leq I$, and all $0 < p < 1$,
\begin{equation}
\label{eqn:that-damn-conjecture}
\Lambda_p(E) \leq 1 - \left( 1 - \frac{1}{d}  \right)^{d-1}\,.
\end{equation}
In particular, $\Lambda_p(E) \leq 1 - 1/e \approx 0.632$.
\end{conjecture}

Note that, by Thm.~\ref{thm:2d-guessing-set}, this holds for $d = 2$.

Given Thm.~\ref{thm:maximum}, this conjecture is simply a bound on $\Lambda_{1/2}(E)$ for any $E$. Noting that $\tr A_{1/2}(E) = 0$, we may consider the problem of maximizing \eqref{eqn:set-measure-formula} subject to the constraint that
\begin{equation}
 \sum_{i = 1}^k \alpha_i + \sum_{j = 1}^m \beta_j = 0.
 \end{equation}
Allowing the extension to the degenerate case, it can be shown that \eqref{eqn:set-measure-formula} achieves a local maximum with $d-1$ positive eigenvalues, or $\alpha_1 = \alpha_2 = ... = \alpha_{d-1} = -\beta_1/(d-1)$. This gives the result of the conjecture. Essentially, this conjecture states that the example in Thm.~\ref{thm:counter-example} is the best possible case for $p = 1/2$.

However, while it can be shown that this is a local maximum, it is not clear that this is a \emph{global} maximum.  Computationally, the difficulty largely stems from the formula for $\Lambda_p(E)$ being distressingly non-linear, and having terms of alternating sign. Intuitively, however, this seems plausible as $d-1$ positive eigenvalues and $1$ negative eigenvalue would create the largest possible subspace for which $\braket{\psi|A_{1/2}(E)|\psi}$ was positive.

The above bounds suggest that, overall, there are limitations to finding experiments that are more reliable than blind guessing in an informationless, assumption-free setting. The above conjecture entails that there is an absolute upper limit to how large a fraction of $\SSS$ any experiment can be more reliable than blind guessing on. If this conjecture is true, then Thm.~\ref{thm:chernoff-theorem} could be greatly improved. Again, all this hinges on the eigenvalue distribution of $A_p(E)$ for $ 0 \leq E \leq I$, which seems difficult to control.


\section{Proofs}\label{sec:proofs}

This section contains proofs of theorems not given in the previous sections, generally for being too long or unenlightening toward the main topic.

\subsection{The Spectrum of $A_p(E)$}\label{sec:proofs:spectrum}

The following result concern the eigenvalues of $A_p(E)$ as defined in \eqref{eqn:ApE-definition}, with $0 \leq E \leq I$ and $0 < p < 1$. These results are necessary for the proofs in Sec.~\ref{sec:proofs:lambda}, but are not directly connected to the existence or non-existence of reliable experiments, and hence have been relegated to this section.

To begin, consider $E$ as a $d$-dimensional Hermitian matrix with respect to the basis $B$. The operators $\diag E$, $I-E$, and $\diag (I - E)$ may be considered similarly. Given that the sum of Hermitian matrices is Hermitian, we have $A_p(E)$ as a Hermitian matrix as well.

Let $\lambda_i(H)$ denote the $i$-th largest eigenvalue of $H$. For $A, B, C$ Hermitian matrices with $A = B + C$, the Ky Fan eigenvalue inequality \cite{Fan49,Mos11} asserts that for any $m \leq d$,
\begin{equation}
\label{eqn:Ky-Fan}
\sum_{i = 1}^m \lambda_i(A) \leq \sum_{i = 1}^m \lambda_i(B) + \sum_{i = 1}^m \lambda_i(C),
\end{equation}
with equality when $m = d$, giving $\tr A = \tr B + \tr C$. 

In our case, let $1 \geq \mu_1 \geq \mu_2 \geq ... \geq \mu_d \geq 0$ be the eigenvalues of $E$, and $1 \geq d_1 \geq d_2 \geq ... \geq d_d \geq 0$ the diagonal entries of $E$---corresponding to the eigenvalues of $\diag E$. The eigenvalues and diagonal entries of $I - E$ are easy to express in terms of $\mu_i, d_i$. The Schur--Horn theorem \cite{Schur23,Horn54,SH} asserts that, for all $m \leq d$,
\begin{equation}
\label{eqn:SchurHornTheorem}
\sum_{i = 1}^m d_i \leq \sum_{i = 1}^m \mu_i
\end{equation}
with equality when $m = d$ (as can be seen by considering $\tr E$). From this, or from the facts that $\mu_1=\max_{\psi\in\SSS} \braket{\psi|E|\psi}$ and $\mu_d=\min_{\psi\in\SSS} \braket{\psi|E|\psi}$, it can also be shown that, for $0\leq E \leq I$,
\begin{equation}\label{eqn:horn-application}
0 \leq \mu_d \leq d_d \leq d_1 \leq \mu_1 \leq 1.
\end{equation}
We use these results to develop some control on the spectrum of $A_p(E)$.

In the case of $p \leq 1/2$, we have that $A_p(E) = p \diag E - (1-p)E$. Applying \eqref{eqn:Ky-Fan} to $B = p \diag E$ and $C = -(1-p)E$,
\begin{equation}
\label{eqn:p-less-than-1/2}
\sum_{i = 1}^m \lambda_i(A_p(E)) \leq p \left(\sum_{i = 1}^m  d_i\right)  - (1-p)\left(\sum_{i = d-m + 1}^d \mu_i \right).
\end{equation}

Similarly, for $p > 1/2$, noting the ordering of the eigenvalues of $I-E$ and $\diag(I-E)$,

\begin{equation}
\label{eqn:p-greater-than-1/2}
\sum_{i = 1}^m \lambda_i(A_p(E)) \leq (1-p)\left(\sum_{i = d-m + 1}^d  (1 - \mu_i)\right)   - p \left(\sum_{i = 1}^m (1-d_i)\right) .
\end{equation}
In the case that $m = d$, these reproduce the trace relation given by \eqref{eqn:trace-equation}.

Using these results, we may derive the following bounds on the spectrum of $A_p(E)$.

\begin{theorem}
Let $0 \leq E \leq I$.\\
For $p \leq 1/2$,
\begin{equation}
-\min(1, \tr E)(1-p)\leq \lambda_d(A_p(E)) \leq \lambda_1(A_p(E)) \leq \min (1, \tr E)p.
\end{equation}
For $p > 1/2$,
\begin{equation}
-\min\left(1, \tr \left[I-E\right]\right) p \leq \lambda_d(A_p(E)) \leq \lambda_1(A_p(E)) \leq \min \left(1, \tr \left[ I - E\right] \right)(1-p).
\end{equation}
\end{theorem}

\begin{proof}
Let $p \leq 1/2$. In that case, taking $m = 1$ in \eqref{eqn:p-less-than-1/2} yields
\begin{equation}
\begin{split}
\lambda_1(A_p(E)) & \leq p d_1  - (1-p)\mu_d \\
& \leq p d_1\,.
\end{split}
\end{equation}
From \eqref{eqn:horn-application}, we have that $d_1 \leq 1$. However, the diagonal entries of $E$ are positive and sum to $\tr E$, hence $d_1 \leq \tr E$. This yields $d_1 \leq \min \left( 1, \tr E \right)$, which completes the upper bound on $\lambda_1(A_p(E))$.

The lower bound on $\lambda_d(A_p(E))$ follows similarly, taking $m = d-1$ in \eqref{eqn:p-less-than-1/2} and noting that
\be
\sum_{i = 1}^d \lambda_i(A_p(E)) = \tr A_p(E) = - (1-2p) \tr E
\ee
and
\be
\sum_{i = 1}^d \mu_i = \sum_{i = 1}^d d_i = \tr E.
\ee

The proof for $p > 1/2$ goes in exactly the same way, using \eqref{eqn:p-greater-than-1/2}.
\end{proof}

We may also use this approach to develop bounds on the sums of the positive eigenvalues and the sums of the negative eigenvalues of $A_p(E)$.

\begin{theorem}\label{eigenvaluesums}
Let $0 \leq E \leq I$. Define $\alpha$ to be the sum of the positive eigenvalues of $A_p(E)$, and $\beta$ to be the sum of the negative eigenvalues.\\
For $p \leq 1/2$,
\begin{equation}
-(1-p) \tr E \leq \beta \leq 0 \leq \alpha \leq p \tr E.
\end{equation}
For $p > 1/2$,
\begin{equation}
-p \tr \left[I-E\right] \leq \beta \leq 0 \leq \alpha \leq (1-p) \tr \left[I-E\right].
\end{equation}
\end{theorem}

\begin{proof}
The upper bounds on $\alpha$ follow from \eqref{eqn:p-less-than-1/2} and \eqref{eqn:p-greater-than-1/2} by dropping the negative terms in the bounds, and bounding the remaining sum from above by the appropriate trace. The lower bounds on $\beta$ follow from the upper bounds on $\alpha$ using $\beta = \tr A_p(E)-\alpha$ and \eqref{eqn:trace-equation}.
\end{proof}

As a quick corollary, recalling \eqref{eqn:trace-equation}, we have that
\begin{corollary}\label{late-night-corollary}\ \\
For $p < 1/2$,
\begin{equation}
\frac{1-p}{1-2p} \tr A_p(E) \leq \beta \leq 0 \leq \alpha \leq -\frac{p}{1-2p} \tr A_p(E).
\end{equation}
For $p > 1/2$,
\begin{equation}
\frac{p}{2p-1} \tr A_p(E) \leq \beta \leq 0 \leq \alpha \leq -\frac{1-p}{2p-1} \tr A_p(E).
\end{equation}
\end{corollary}

\subsection{Computing and Bounding $\Lambda_p(E)$}\label{sec:proofs:lambda}

For a general Hermitian $A$, Thm.~\ref{thm:measure-formula} gives an explicit formula for the measure of the set of $\psi \in \SSS$ for which $\braket{\psi|A|\psi} > 0$. Letting $\alpha_1 \geq \alpha_2 \geq ... \geq \alpha_k > 0$ be the positive eigenvalues of $A$, and $0 \geq \beta_1 \geq \beta_2 \geq ... \geq \beta_m$ the negative, assuming the non-degenerate case of $\alpha_i \neq \alpha_j$ for $i \neq j$, the claim was that
\begin{equation}
\mu \Bigl\{ \psi \in \SSS : \braket{\psi|A|\psi} > 0 \Bigr\} = \sum_{i = 1}^k \frac{\alpha_i^{d-1}}{\prod_{h = 1}^m(\alpha_i - \beta_h) \prod_{j = 1, j \neq i}^k(\alpha_i - \alpha_j)}.
\end{equation}

\begin{proof}[Proof of Thm.~\ref{thm:measure-formula}]
Taking $\mu$ as the normalized uniform measure on $\SSS$, we may interpret this quantity probabilistically as the probability that $\braket{\psi|A|\psi} > 0$ for $\psi$ sampled uniformly from $\SSS$. We may sample from $\SSS$ uniformly by normalizing a vector sampled from the Gaussian distribution over $\mathbb{C}^d$, with mean $0$ and covariance operator $I$. Taking $\psi$ as sampled uniformly and $\phi$ as sampled with the standard normal Gaussian distribution on $\mathbb{C}^d$,
\begin{equation}
\begin{split}
\mu \Bigl\{ \psi \in \SSS : \braket{\psi|A|\psi} > 0 \Bigr\} & = \PPP\bigl( \braket{\psi|A|\psi} > 0 \bigr) \\
& = \PPP\left( \frac{\braket{ \phi|A|\phi }}{ \|\phi\|^2 } > 0 \right) \\
& = \PPP\bigl( \braket{\phi|A|\phi} > 0 \bigr).
\end{split}
\end{equation}

Let $u_i$ be the eigenvector corresponding to $\alpha_i$, and $v_j$ the eigenvector corresponding to $\beta_j$. We may use the eigenvalues and vectors of $A$ to compute the above probability precisely.

This Gaussian distribution is invariant under unitary transformations. Thus, if $\phi$ is sampled from this distribution, its components, the $\braket{u_i | \phi}, \braket{v_j | \phi}$, are independent random variables with the complex Gaussian distribution, with mean $0$ and variance $1$, i.e., the real and imaginary parts are independent real Gaussian random variables with mean $0$ and variance $1/2$. 

Let $A_i = \bigl|\braket{\phi|u_i}\bigr|^2, B_j = \bigl|\braket{\phi|v_j}\bigr|^2$. As a consequence of sampling from the Gaussian distribution, the $A_i, B_j$ are independent random variables with identical exponential distributions, $Exp(1/2)$.
Continuing,
\begin{equation}\label{eqn:im-so-tired}
\begin{split}
\mu \{ \psi \in \SSS : \braket{\psi|A|\psi} > 0 \} &= \PPP( \braket{\phi|A|\phi} > 0 ) \\
&= \PPP \left( \sum_{i = 1}^k \alpha_i |\braket{\phi|u_i}|^2 + \sum_{j = 1}^m \beta_j |\braket{\phi|v_j}|^2 > 0\right) \\
&= \PPP \left( \sum_{i = 1}^k \alpha_i A_i >  \sum_{j = 1}^m -\beta_j B_j \right) \,.
\end{split}
\end{equation}
This probability may be calculated explicitly, as we know the distributions of the $A_i$ and the $B_j$. First, we must derive the distributions of the respective sums. For a set of values $\lambda_i > 0$, define $S_n = \sum_{i = 1}^n \lambda_i X_i$ where the $X_i$ are i.i.d., $X_i \sim Exp(1/2)$. Note, $S_n$ is of the same form as the two sums in \eqref{eqn:im-so-tired}. Defining $f_n(c)$ as the probability density function of $S_n$, the recursive structure $S_n = S_{n-1} + \lambda_n X_n$ leads to
\begin{equation}
\begin{split}
f_1(c) & = \frac{1}{2 \lambda_1} e^{-\frac{c}{2 \lambda_1} }\,,\\
f_{n+1}(c) & = \frac{1}{2 \lambda_{n+1} } e^{-\frac{c}{2 \lambda_{n+1} } } \int_{0}^c f_{n}(s) \, e^{\frac{s}{2 \lambda_{n+1} }} \, ds \,.
\end{split}
\end{equation}
If the $\lambda_i$ are mutually distinct, then this recursion has the explicit solution
\begin{equation}\label{solution1}
f_n(c) = \sum_{i = 1}^n \frac{ e^{-\frac{c}{2 \lambda_i} } \lambda_i^{n-2} }{2 \prod_{j = 1, j \neq i}^{n} (\lambda_i - \lambda_j ) }\,,
\end{equation}
as one can verify by a calculation using the identity \cite[Eq.~(5)]{CZ14}
\begin{equation}\label{identity1}
\sum_{i=1}^n \frac{\lambda_i^{n-2}}{\prod_{j=1,j\neq i}^n (\lambda_i-\lambda_j)} =0\,,
\end{equation}
which follows from the Lagrange interpolation formula. The latter asserts that, for any polynomial $P(x)$ of degree $\leq n-1$ and any  distinct numbers $x_1,\ldots,x_n$,
\begin{equation}\label{Lagrange}
P(x) = \sum_{i=1}^n P(x_i) \prod_{j=1,j\neq i}^n \frac{x-x_j}{x_i-x_j}\,.
\end{equation}
(The Lagrange interpolation formula follows from the facts that both sides agree for $x=x_i$, and that two polynomials of degree $\leq n-1$ that agree on $n$ points agree everywhere.) For $P(x)=x^{n-2}$ and $x_i=\lambda_i$, the coefficient of $x^{n-1}$ in \eqref{Lagrange} yields \eqref{identity1}.

We may use \eqref{solution1} to give the densities for $\sum_{i = 1}^k \alpha_i A_i$ and $\sum_{j = 1}^m -\beta_j B_j$ explicitly in terms of the $\alpha_i, \beta_j$, provided the $\beta_j$ are mutually distinct. (Recall that we assumed in Thm.~\ref{thm:measure-formula} that the $\alpha_i$ are distinct, but we did not assume this for the $\beta_j$.) Densities in hand, computing the probability \eqref{eqn:im-so-tired}, and hence $\mu \{ \psi \in \SSS : \braket{\psi|A|\psi} > 0 \}$, is just a matter of integration and yields 
\begin{equation}\label{probAB}
\sum_{i=1}^k \sum_{h=1}^m \frac{\alpha_i^k (-\beta_h)^{m-1}}{(\alpha_i-\beta_h) \prod_{j=1,j\neq i}^k (\alpha_i-\alpha_j) \prod^m_{\ell=1,\ell\neq h} (\beta_\ell-\beta_h)} \,.
\end{equation}
By the Lagrange interpolation formula \eqref{Lagrange} again, this time for $n=m$, $P(x)=x^{m-1}$, $x=\alpha_i$, and $x_i=\beta_h$, we have that 
\begin{equation}
\alpha_i^{m-1} = \sum_{h=1}^m \beta_h^{m-1} \prod_{\ell=1,\ell\neq h}^m \frac{\alpha_i-\beta_\ell}{\beta_h-\beta_\ell}\,,
\end{equation}
and in the light of this identity, \eqref{probAB} can be simplified to
\begin{equation}\label{formula1}
\sum_{i=1}^k \frac{\alpha_i^k \alpha_i^{m-1}}{\prod_{j=1,j\neq i}^k (\alpha_i-\alpha_j)  \prod_{\ell=1}^m (\alpha_i-\beta_\ell)} \,,
\end{equation}
which yields \eqref{eqn:set-measure-formula}. Since the expression \eqref{eqn:im-so-tired} is a continuous function of the $\beta_j$, the formula \eqref{formula1}, and hence \eqref{eqn:set-measure-formula}, apply also in the limiting case when some of the $\beta_j$ coincide.
\end{proof}

We may use this probabilistic approach to give a proof of Thm.~\ref{thm:chernoff-theorem}, that $\Lambda_p(E) \leq 4p(1-p)$.

\begin{proof}[Proof of Thm.~\ref{thm:chernoff-theorem}]
For $p=1/2$ or $p=0$ or $p=1$ or $E=0$ or $E=I$, the inequality is trivially true, so let us assume that none of these is true of the given $p$ and $E$ with $0\leq E\leq I$. Letting 
$\alpha_1 \geq \alpha_2 \geq ... \geq \alpha_k > 0 \geq \beta_1 \geq \beta_2 \geq ... \geq \beta_m$ be the eigenvalues of $A_p(E)$, from \eqref{eqn:im-so-tired},
\begin{equation}
\Lambda_p(E) = \PPP \left( \sum_{i = 1}^k \alpha_i A_i + \sum_{j = 1}^m \beta_j B_j> 0 \right).
\end{equation}
For any fixed value $t$ with $0<t <(2\sum_{i=1}^k \alpha_i)^{-1}$, the above may be rewritten as
\begin{equation}
\Lambda_p(E) = \PPP \left( e^{(\sum_{i = 1}^k \alpha_i A_i + \sum_{j = 1}^m \beta_j B_j)t }> 1 \right),
\end{equation}
and Markov's inequality gives the bound
\begin{equation}
\label{eqn:markov-bounds}
\begin{split}
\Lambda_p(E) & \leq \EEE\left[ e^{(\sum_{i = 1}^k \alpha_i A_i + \sum_{j = 1}^m \beta_j B_j)t } \right] \\
& = \EEE\left[ \prod_{i = 1}^k e^{\alpha_i A_i t} \prod_{j = 1}^m e^{\beta_j B_j t} \right] \\
& = \prod_{i = 1}^k \EEE\left[e^{\alpha_i A_i t}\right] \prod_{j = 1}^m \EEE\left[e^{\beta_j B_j t} \right] \\
& = \prod_{i = 1}^k \frac{1}{1 - 2 \alpha_i t} \prod_{j = 1}^m \frac{1}{1 - 2 \beta_j t}\,.
\end{split}
\end{equation}
The last step used that $2\alpha_i t\leq 2\sum_i \alpha_i t<1$.
Now note that, for $0<x<y<1$, $(1-x)(1-y)$ shrinks if we replace $x\to x-dx$, $y\to y+dx$ with infinitesimal $dx>0$. Thus, when $x+y=c\in [0,1]$ is fixed, $(1-x)(1-y)$ is minimized on $[0,1]^2$ at $x=c,y=0$ or $x=0,y=c$. It follows further that $\prod_{i=1}^k (1-x_i)$, when $\sum_{i=1}^k x_i=c\in [0,1]$ is fixed, is minimized on $[0,1]^k$ by $x_1=c,x_2=\ldots=x_k=0$ (or permutations thereof), and the minimal value is $1-c$. As a consequence for $x_i=2\alpha_it$,
\be
\prod_{i=1}^k \frac{1}{1-2\alpha_it} \leq \frac{1}{1-2\alpha t}
\ee
with $\alpha = \sum_{i = 1}^k \alpha_i$. From this and \eqref{prod1+xi} for $x_j=-2\beta_j t$, setting $\beta = -\sum_{j = 1}^m \beta_j$ (note the negative), we obtain that 
\begin{equation}
\Lambda_p(E) \leq \frac{1}{1 - 2 \alpha t} \frac{1}{1 + 2 \beta t}.
\end{equation}
Fixing $\alpha, \beta$, we may choose $t$ to minimize this bound, namely $t = (\beta - \alpha)/(4 \alpha \beta)$ (which indeed lies between 0 and $(2\alpha)^{-1}$ because $\alpha-\beta=\tr A_p(E)<0$ and $\alpha,\beta>0$). This yields an upper bound of
\begin{equation}
\Lambda_p(E) \leq \frac{4 \alpha \beta}{ (\alpha + \beta)^2 }.
\end{equation}
We have that $\alpha = \beta + \tr A_p(E)$, and
\begin{equation}
\Lambda_p(E) \leq \frac{4 \beta( \beta + \tr A_p(E))}{ (2\beta + \tr A_p(E))^2 }.
\end{equation}
This bound is maximized taking the largest possible value of $\beta$. From Thm.~\ref{eigenvaluesums}, (noting the sign difference in $\beta$s), and recalling what the trace of $A_p(E)$ is, we recover the desired bound $\Lambda_p(E) \leq 4p(1-p)$.
\end{proof}

Finally, we present the proof of Thm.~\ref{thm:one-negative-eigenvalue}, which asserts that, for $p < \frac{\ln 2}{1 + \ln 2}$ and $p > \frac{1}{1 + \ln 2}$, for any $0 \leq E \leq I$ for which $A_p(E)$ has only one negative eigenvalue, $\Lambda_p(E) \leq 1/2$.

\begin{proof}[Proof of Thm.~\ref{thm:one-negative-eigenvalue}]
Fix $E$ such that the eigenvalues of $A_p(E)$ are $\alpha_1 \geq \alpha_2 \geq ... \geq \alpha_{d-1} \geq 0 > \beta$. Instead of applying \eqref{eqn:set-measure-formula} directly, it is more convenient to apply it to $A = -A_p(E)$, in which case the signs of the eigenvalues are flipped. The measure for $A_p(E)$ and the measure for $-A_p(E)$ sum to 1. Hence, 
\begin{equation}
\label{eqn:one-negative-eigenvalue-formula}
\begin{split}
\Lambda_p(E) & = 1 - \frac{ \beta^{d-1} }{ \prod_{i = 1}^{d-1} (\beta - \alpha_i) }\\
& = 1 - \frac{1}{\prod_{i = 1}^{d-1} (1 - \alpha_i/\beta)}.
\end{split}
\end{equation}
It suffices to show that $\prod_{i = 1}^{d-1} (1 - \alpha_i/\beta) \leq 2$. The arithmetic-geometric mean property gives
\begin{equation}
\label{eqn:arith-geo-mean}
\begin{split}
\left[\prod_{i = 1}^{d-1} (1 - \alpha_i/\beta)\right]^{1/(d-1)} & \leq \frac{1}{d-1}\left(d-1 - \frac{1}{\beta}\left(\sum_{i = 1}^{d-1} \alpha_i \right) \right) \\
& = 1 - \frac{1}{d-1} \frac{\tr A_p(E) - \beta}{\beta} \\
& = 1 + \frac{1}{d-1} \left(1 - \frac{\tr A_p(E)}{\beta}\right) \\
\end{split}
\end{equation}
From Cor.~\ref{late-night-corollary}, we have that for $p < 1/2$ (and also for $p=1/2$ since then $\tr A_p(E)=0$),
\begin{equation}
1- \tr A_p(E)/\beta \leq \frac{p}{1-p},
\end{equation}
and for $p > 1/2$,
\begin{equation}
1-  \tr A_p(E)/\beta \leq \frac{1-p}{p}.
\end{equation}
Using the standard fact that $(1+\frac{x}{n})^n$ converges monotonically increasingly to $e^x$ for $x>0$,
\begin{equation}
\begin{split}
\prod_{i = 1}^{d-1} (1 - \alpha_i/\beta) & \leq \left[ 1 + \frac{1}{d-1} \left(1 - \frac{\tr A_p(E)}{\beta}\right) \right]^{d-1} \\
& \leq e^{ 1 - \tr A_p(E)/\beta } \\
& \leq \begin{cases}
e^{ p/(1-p) } & \text{if }p \leq 1/2 \\
e^{ (1-p)/p } & \text{if }p > 1/2.
\end{cases}
\end{split}
\end{equation}
Taking $p < \ln 2/(1 + \ln 2)$ or $p > 1/(1 + \ln 2)$ forces $\prod_{i = 1}^{d-1} (1 - \alpha_i/\beta) \leq 2$ for all $d \geq 2$. Hence, $\Lambda_p(E) \leq 1/2$.
\end{proof}

\bigskip

\textit{Acknowledgments.}
Both authors are supported in part by NSF Grant SES-0957568. 
R.T.~is supported in part by grant no.~37433 from the John Templeton Foundation and by the Trustees Research Fellowship Program at Rutgers, the State University of New Jersey.


\begin{thebibliography}{29}

\bibitem{CZ14} M. Chamberland, D. Zeilberger:
	A Short Proof of McDougall's Circle Theorem.
	\textit{American Mathematical Monthly} \textbf{121(3)}: 263--265 (2014)

\bibitem{CT12a} C. W. Cowan, R. Tumulka: 
	Can One Detect Whether a Wave Function Has Collapsed?
	\textit{Journal of Physics A: Mathematical and Theoretical} \textbf{47}: 195303 (2014)
	\url{http://arxiv.org/abs/1307.0810}

\bibitem{CT12b} C. W. Cowan, R. Tumulka: 
	Epistemology of Wave Function Collapse in Quantum Physics. 
	\textit{British Journal for the Philosophy of Science} online first (2014)
	\url{http://arxiv.org/abs/1307.0827}

\bibitem{DGZ04} D. D\"urr, S. Goldstein, N. Zangh\`\i: 
	Quantum Equilibrium and the Role of Operators as 
	Observables in Quantum Theory. 
	\textit{Journal of Statistical Physics} \textbf{116}: 959--1055 (2004)
	 \url{http://arxiv.org/abs/quant-ph/0308038}

\bibitem{Fan49} K. Fan: 
	On a theorem of Weyl concerning eigenvalues of linear transformations I.
	\textit{Proceedings of the National Academy of Sciences USA}
	\textbf{35}: 652--655 (1949)

\bibitem{Horn54} A. Horn:
	Doubly stochastic matrices and the diagonal of a rotation matrix.
	\textit{American Journal of Mathematics} \textbf{76}: 620--630  (1954)

\bibitem{GRW86} G.C. Ghirardi, A. Rimini, T. Weber: 
	Unified Dynamics for Microscopic and Macroscopic Systems. 
	\textit{Physical Review D} \textbf{34}: 470--491 (1986)

\bibitem{GTZ07} S. Goldstein, R. Tumulka, N. Zangh\`\i:
  The Quantum Formalism and the GRW Formalism.
  \textit{Journal of Statistical Physics} \textbf{149}: 142--201 (2012)
   \url{http://arxiv.org/abs/0710.0885}

\bibitem{Hel76} C. W. Helstrom: 
	\textit{Quantum Detection and Estimation Theory.}
	New York: Academic Press (1976)

\bibitem{Mos11} M. S. Moslehian:
	Ky Fan Inequalities.
	\textit{Linear and Multilinear Algebra} \textbf{60(11-12)}: 1313--1325 (2012)
	\url{http://arxiv.org/abs/1108.1467}

\bibitem{Schur23} I. Schur: 
	\"Uber eine Klasse von Mittelbildungen mit Anwendungen auf die Determinantentheorie.
	\textit{Sitzungsberichte der Berliner Mathematischen Gesellschaft} \textbf{22}: 9--20 (1923)

\bibitem{SH} Schur--Horn theorem. 
	In \textit{Wikipedia, the free encyclopedia} (accessed 24 June 2013) 
	\url{http://en.wikipedia.org/wiki/Schur-Horn_theorem}

\end{thebibliography}
\end{document}